\documentclass[11pt,letterpaper]{amsart}
\usepackage{graphicx}    
\usepackage{epsfig,amsmath,euscript,amssymb}

\usepackage{natbib}

\usepackage{amsthm}
\usepackage{color}
\newtheorem{proposition}{Proposition}
\newtheorem{remark}{Remark}

\begin{document}

\title[Sensitivities of prices of Asian options]
{Sensitivities of Asian options in the Black-Scholes model}

\author{Dan Pirjol}
\address
{School of Business\newline
\indent Stevens Institute of Technology\newline
\indent Hoboken, NJ 07030\newline
\indent United States of America}
\email{
dpirjol@gmail.com}

\author{Lingjiong Zhu}
\address
{Department of Mathematics\newline
\indent Florida State University\newline
\indent 1017 Academic Way\newline
\indent Tallahassee, FL-32306\newline
\indent United States of America}
\email{
zhu@math.fsu.edu}

\date{31 January 2018}

\begin{abstract}
We propose analytical approximations for the sensitivities (Greeks) of the Asian 
options in the Black-Scholes model, following from a small maturity/volatility
approximation for the option prices which has the exact short
maturity limit, obtained using large deviations theory. Numerical tests 
demonstrate
good agreement of the proposed approximation with alternative numerical 
simulation results for cases of practical interest.
We also study the qualitative properties of Asian Greeks, including new results 
for Rho, the sensitivity with respect to changes in the risk-free rate, and 
Psi, the sensitivity with respect to the dividend yield. In particular
we show that the Rho of a fixed-strike Asian option and the Psi of a 
floating-strike Asian option can change sign. 
\end{abstract}

\maketitle

\section{Introduction}

The sensitivities of derivatives prices with respect to changes in market
risk factors, also known as Greeks, play an important role in the hedging 
and risk management of these products. This requires a good understanding of 
their properties and qualitative behavior under changes in the model parameters. 
For this reason a great deal of
effort has been put into developing precise computational methods for 
numerical evaluation of sensitivities. The simplest method makes use of finite
difference using exact pricing (when available), or numerical approximations
such as Monte Carlo simulation \citep{PGbook}. Pathwise estimators for the
sensitivities of vanilla
and path-dependent derivatives have been presented in \cite{BrGl}.
More sophisticated methods have 
been also proposed to improve the efficiency and speed of the sensitivities
calculation, using Malliavin calculus \citep{Mall1}, and adjoint (algorithmic)
differentiation \citep{adj}.

While the sensitivities of the vanilla European options are well understood,
the corresponding parameters for exotic derivatives are more challenging, due 
to the absence of analytical expressions for their prices. This is the case
with the Asian options, for which no analytical result is available, even in the
simple Black-Scholes model. 
The pricing of Asian options with continuous-time averaging in this model
is related to the distributional properties for the time integral of geometric 
Brownian motion. This is known in closed form, although the result is 
very complicated and requires a delicate numerical integration, see \cite{Yor}, \cite{Dufresne}. 
These results have been applied to the
pricing of Asian options in \cite{GY}, \cite{CS}, 
see \cite{DufresneReview} for an overview. 
These results are not easy to use for the sensitivity analysis.

In the literature, various methods have been proposed for the 
numerical evaluation of the Asian option prices: the inversion of a Laplace
transform method of \cite{GY}, the spectral method in \cite{Linetsky},
the analytical approximation of the solution of a pricing PDE in \cite{FPP}, 
and the short maturity 
approximation in \cite{PZAsian} and by \cite{ALW}. For practical applications 
the lower analytical bounds in \cite{Curran} based on conditioning on the
geometric average, and the upper bounds of \cite{RS} and based on the 
co-monotonic approximation \cite{VDLDG} give strong constraints on the Asian 
options prices. 

The discrete approximation of the underlying
process as a Continuous Time Markov Chain (CTMC) has been proposed in \cite{Cai}
for pricing Asian options, and has been simplified in \cite{CLL} such that
it requires the inversion of a single Laplace transform. 
The CTMC  method can be also used for pricing Asian options in the Black-Scholes
model by discretizing a Markov process which has the same terminal distribution
as the time integral of the geometric Brownian motion \citep{4}. This method has the advantage
that no Laplace transform inversion is required.

The Greeks of Asian
options under the Black-Scholes model have been studied in \cite{BP,Benh,CK}, 
and a discussion under the local volatility model was presented in \cite{FPP}.
Pathwise and likelihood estimators for the Greeks of the Asian options 
under the BS model have been given in \cite{BrGl}.

The signs of the Asian sensitivities are also relevant for practical
applications. Their signs give the qualitative dependence 
of the Asian option price on the relevant risk factors, even when a more 
precise determination is not available. 
The sign of Gamma determines the PnL of a Delta hedging strategy under 
misspecified volatility, see \cite{EKJS}. Derivatives with positive definite
Gamma have a robustness property under hedging. 
The positivity of Vega implies the monotonicity of the Asian option prices in
the maturity argument, which yields a non-arbitrage argument similar to the
calendar arbitrage for vanilla options.

There have been some rigorous studies for the qualitative behavior of the 
sensitivities for Asian options. 
It was proved in \cite{Carr} that 
the Vega is always positive for the Black-Scholes model for the Asian call 
options. This property does not always hold in a 
binomial tree approximation of the same model, see \cite{Carr}. 
\cite{Carr} also showed that the Asian call option price increases
with respect to the maturity $T$. 

In this paper, we present two results for the sensitivities of the Asian options
in the Black-Scholes model:

(i) We give analytical approximations for the sensitivities of an Asian
option in the Black-Scholes model. They follow from a short maturity asymptotics
of the Asian prices proposed in \cite{PZAsian}, which has the correct small 
time asymptotics derived using large deviations theory for diffusion
processes \citep{Varadhan67}. Both fixed-strike and floating-strike Asian options
are considered.

(ii) We study rigorously the signs of the Asian sensitivities in the 
Black-Scholes model. The positivity of Vega has been proved in \cite{Carr}.
We discuss all other sensitivities, including Rho, the sensitivity 
of the Asian option with respect to the risk-free rate $r$. We show that,
unlike the European option, which has a strictly positive Rho, the Asian
option Rho can change sign. A similar result holds for $\Psi$,
which is the sensitivity of the price with respect to the dividend yield.

The Black-Scholes model and its variations are commonly used for risk 
management of Asian options in practical applications, see for example 
\cite{Geman} for an overview for Asian options on commodities.
In commodities, the underlyings
of these options are futures contracts, and the simplest model is a 
multi-factor Black-type model with time-dependent volatility, chosen such as 
to satisfy the Samuelson rule. This reduces to the Black model for futures 
in the limit of one factor and time-independent volatility.

Asian options on FX rates are also traded in practice. 
The simplest model for FX rates dynamics assumes that they follow geometric 
Brownian motion with a drift proportional 
to the difference of the risk-free rates in the two currencies.
Denoting $X_t$ the FX rate expressed as the number of domestic currency units
corresponding to one foreign currency unit, under the log-normal model one has
$dX_t = \sigma_{\rm FX} X_t dW_t + (r_d-r_f) X_t dt$, where $r_d$ is the risk-free
rate of the domestic currency, $r_f$ is the risk-free rate of the foreign currency, 
and $\sigma_{\rm FX}$ is the volatility. 
The sensitivity to the rates $r_d,r_f$ is mathematically equivalent to Rho and 
Psi for the risk-free rate and the dividend yield, and it is an important risk 
factor for FX Asian options.

We organize the paper as follows. In Section \ref{ctsSection}, we study the 
continuous-time Black-Scholes model, and present a short maturity analytical
approximation for the Greeks of the Asian options. We demonstrate good
numerical agreement with alternative numerical evaluations of the sensitivities.
In Section \ref{Sec:Qualitative}, we review some qualitative results
for the Asian Greeks, and present some new results for Rho and Psi of the 
fixed-strike Asian 
options in Section \ref{ctsFixed} and floating-strike Asian options in Section \ref{ctsFloating}.


\section{Analytical approximations for Asian Greeks}\label{ctsSection}

In the Black-Scholes model, the asset price $(S_{t})_{t\geq 0}$ follows the dynamics
\begin{equation}
dS_{t}=(r-q)S_{t}dt+\sigma S_{t}dW_{t},\qquad S_{0}>0,
\end{equation}
where $W_{t}$ is a standard Brownian motion started from $0$ at time $0$, 
$r$ is the risk-free rate,
$q$ is the dividend yield and $\sigma\geq 0$ is the volatility. 

\subsection{Analytical approximations for the Greeks of fixed-strike Asian options}\label{ShortMat}

Let $T$ be the maturity and $K$ be the fixed-strike price. 
The price of the fixed-strike call and put Asian options are given by
\begin{align}
&C=e^{-rT}\mathbb{E}\left[\left(\frac{1}{T}\int_{0}^{T}S_{t}dt-K\right)^{+}\right],
\\
&P=e^{-rT}\mathbb{E}\left[\left(K-\frac{1}{T}\int_{0}^{T}S_{t}dt\right)^{+}\right].
\end{align}

In \cite{PZAsian} a short maturity  approximation has been proposed 
for the prices of Asian options in the local volatility model. A similar
result was obtained in \cite{ALW}.
We will state the result for the simpler case when the underlying is assumed
to follow the Black-Scholes model.
In the Black-Scholes model, ignoring the interest rate and dividend yield,
by the Brownian scaling property, 
the dependence of the option prices on the volatility and maturity enters only
through the realized variance $\sigma^{2}T$, such that 
the short maturity approximation in \cite{PZAsian} is equivalent to 
the small realized variance limit. 
Thus a more appropriate characterization of our limit is as the 
small maturity/volatility approximation, which is relevant in practice
since the magnitude of the realized variance is often small, even though 
the maturity may not be.

The approximation proposed in \cite{PZAsian} for the Asian option prices
has a form similar to the Black-Scholes formula
\begin{align}\label{CAsian}
&\overline{C} = e^{-rT} [A(T) N(d_1) - K N(d_2) ]\,, \\
\label{PAsian}
&\overline{P} = e^{-rT} [KN(-d_2) - A(T) N(-d_1)]\,,
\end{align}
where 
\begin{equation}
A(T) = S_0 \frac{1}{(r-q)T} (e^{(r-q)T} - 1)\,, 
\end{equation}
and
\begin{equation}
d_{1,2}=\frac{1}{\Sigma_{\rm LN}(K/S_0) \sqrt{T}} 
\left(\log\frac{A(T)}{K} \pm \frac12 
\Sigma^2_{\rm LN}(K/S_0) T \right)\,,
\end{equation}

Requiring that the approximation (\ref{CAsian}), (\ref{PAsian}) has the
correct $T\to 0$ asymptotics determines $\Sigma_{\rm LN}(K/S_0)$.
This is a calculable function given explicitly
in Proposition 18 of \cite{PZAsian}, for the more general case of the
local volatility model, including the Black-Scholes model as a special case 
of the constant volatility (Proposition 12 of \cite{PZAsian}):
\begin{equation}
\Sigma_{\rm LN}^{2}(K/S_{0})=\frac{\sigma^{2}\log^{2}(K/S_{0})}{2\mathcal{J}_{\rm BS}(K/S_{0})},
\end{equation}
where
\begin{equation}
\mathcal{J}_{\rm BS}(K/S_{0})
=
\begin{cases}
\frac{1}{2}\beta^{2}-\beta\tanh(\beta/2) &\text{for $K\geq S_{0}$},
\\
2\xi(\tan\xi-\xi) &\text{for $K\leq S_{0}$},
\end{cases}
\end{equation}
where $\beta\in[0,\infty)$ is the unique solution of the equation 
\begin{equation}
\frac{\sinh(\beta)}{\beta}=\frac{K}{S_{0}}\,,
\end{equation}
and $\xi\in[0,\pi/2)$ is the unique solution of the equation 
\begin{equation}
\frac{\sin(2\xi)}{2\xi}=\frac{K}{S_{0}}\,.
\end{equation}

The function $\Sigma_{\rm LN}(K/S_0)$ is the short maturity limit of the
so-called equivalent log-normal volatility of the Asian option. 
This is defined as that 
volatility of an European option with the same maturity $T$ and strike $K$
as those of the Asian option, which reproduces the Asian option price when
used in the Black-Scholes formula, as in (\ref{CAsian}), (\ref{PAsian}). 
We note that an alternative approximation for the short maturity Asian
option prices can be given in terms of a Bachelier formula, with an 
equivalent normal volatility $\Sigma_{\rm N}(K, S_0)$. The short maturity
asymptotics of $\Sigma_{\rm N}(K, S_0)$ is given in Proposition 18 (ii) of
\cite{PZAsian}. 

The main properties of this function are:

(i) The function $\Sigma_{\rm LN}(K/S_0)$ is linear in $\sigma$. The
ratio $\Sigma_{\rm LN}(K/S_0)/\sigma$ depends only on $K/S_0$. 
The plot of this function is shown in Figure~\ref{Fig:SigLN}.

(ii) The equivalent log-normal volatility of the Asian option in the 
Black-Scholes model has the expansion around the ATM point
\begin{equation}\label{SigLNTaylor}
\Sigma_{\rm LN}(K/S_0) = \frac{\sigma}{\sqrt3} \left( 1 + \frac{1}{10}
x - \frac{23}{2100} x^2 + \frac{1}{3500} x^3 + O\left(x^4\right) \right)\,,
\end{equation}
with $x=\log(K/S_0)$, see Eq.~(55) in \cite{PZAsian}. This is shown
as the dashed curve in Figure~\ref{Fig:SigLN}, and the $x=0$ value is
shown as the red line. From this figure we see that the series expansion 
(\ref{SigLNTaylor}) gives a good approximation for $|x| \leq 2$ which
corresponds to the most commonly encountered strikes in practice. 

\begin{figure}
    \centering
   \includegraphics[width=4.0in]{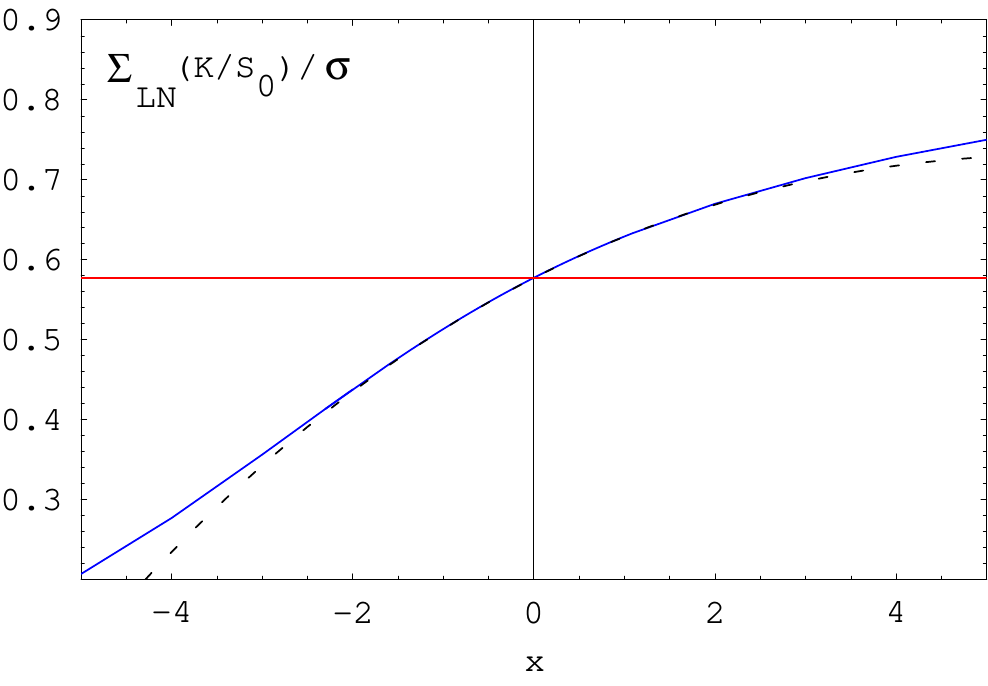}
    \caption{
Plot of $\Sigma_{\rm LN}(K/S_0)/\sigma$ vs $x = \log(K/S_0)$ (blue curve).
The dashed curve shows the series expansion (\ref{SigLNTaylor}) around the 
ATM point $x=0$. The red line at $1/\sqrt{3}$ corresponds to the ATM 
point.}
\label{Fig:SigLN}
 \end{figure}

(iii) As $x\to \infty$, we have $\Sigma_{\rm LN}(K/S_0)/\sigma \to 1$
as 
\begin{equation}
\Sigma_{\rm LN}(K/S_0)= \sigma \left( 1 - \frac{1}{x} \left(\log(2x)-1\right) +O\left(x^{-2}\right)
\right)\,, \quad x\to \infty\,.
\end{equation}

(iv) As $x\to -\infty$, we have $\Sigma_{\rm LN}(K/S_0) \to 0$, as
\begin{equation}
\Sigma_{\rm LN}(K/S_0)= \frac12 \sigma |x| e^{-|x|} \left(1 + O\left(e^{-|x|}\right)\right)
\,, \quad x \to -\infty\,.
\end{equation}

The asymptotics (iii) and (iv) follow from Proposition 13 in \cite{PZAsian}.

The Asian option prices satisfy the put-call parity relation
\begin{equation}\label{PutCall}
C - P = \overline{C} - \overline{P} = e^{-rT} (A(T)-K) \,.
\end{equation}
It is easy to see that this exact result is satisfied by the short 
maturity approximations (\ref{CAsian}).

Using the short maturity approximation (\ref{CAsian}), (\ref{PAsian}) one can
derive simple approximations for the Greeks of the fixed-strike Asian options.
These are not rigorous short maturity limits for the Asian option sensitivities.
However, numerical testing will show that they give reasonably good 
approximations for these
sensitivities by comparing with alternative pricing methods.

\textbf{Delta}

The \textit{Delta} of the approximated Asian call/put option is defined as 
\begin{equation}
\Delta_{\overline{C}} = \left( \frac{\partial \overline{C}}{\partial S_0} \right)\,,\quad
\Delta_{\overline{P}} = \left( \frac{\partial \overline{P}}{\partial S_0} \right)\,.
\end{equation}
From the put-call parity relation (\ref{PutCall}) they are related as
\begin{equation}\label{DeltaPutCall}
\Delta_{\overline{C}} - \Delta_{\overline{P}} = 
e^{-rT} \frac{1}{(r-q)T} \left( e^{(r-q)T}-1\right)\,.
\end{equation}

Using (\ref{CAsian}) we get the explicit result
\begin{equation}\label{DeltaC}
\Delta_{\overline{C}} = e^{-rT} \frac{A(T)}{S_0} \left( N(d_1) - \sqrt{\frac{T}{2\pi}}
e^{-\frac12 d_1^2} \Sigma'_{\rm LN}(K/S_0) \right)\,,
\end{equation}
where
\begin{equation}
\Sigma'_{\rm LN}(K/S_0) := \frac{\partial}{\partial\log(K/S_0)} 
\Sigma_{\rm LN}(K/S_0) \,.
\end{equation}
The first term in (\ref{DeltaC}) gives the contribution due to $S_0$ dependence
of $A(T)$, and the second term is due to the $S_0$ dependence in 
$\Sigma_{\rm LN}(K/S_0)$. For $K/S_0 \sim 1$ around the ATM point, using
(\ref{SigLNTaylor}) we have the expansion
\begin{equation}
\Sigma'_{\rm LN}(K/S_0) := \frac{\sigma}{\sqrt3} \left( \frac{1}{10}
- \frac{23}{1050} x + \frac{3}{3500} x^2 + O(x^3) \right)\,.
\end{equation}

\textbf{Gamma}

The \textit{Gamma} of the approximated Asian call option is
\begin{equation}
\Gamma_{\overline{C}}  : = 
\left( \frac{\partial^2 \overline{C}}{\partial S_0^2} \right)\,.
\end{equation}
From Eq.~(\ref{DeltaPutCall}) one can see that this is the same for
the call and put Asian options.

This can be computed in explicit form and is given by
\begin{align}\label{GammaC}
\Gamma_{\overline{C}} 
&= e^{-rT} \frac{A(T)}{S_0} \frac{1}{\Sigma_{\rm LN}(K/S_0) \sqrt{T} S_0}
\sqrt{\frac{1}{2\pi}} e^{-\frac12 d_1^2}\\
&\qquad\cdot
\left\{\left( 1 + \sqrt{T} \Sigma'_{\rm LN}(K/S_0) \right)
\left( 1 + \frac{\Sigma'_{\rm LN}(K/S_0)}{\Sigma_{\rm LN}(K/S_0)} 
\log\frac{A(T)}{S_0} \right)
+ 
\Sigma_{\rm LN}(K/S_0) \Sigma''_{\rm LN}(K/S_0) T \right\}\,.\nonumber
\end{align}

Here
\begin{equation}
\Sigma''_{\rm LN}(K/S_0) := \frac{\partial^2}{\partial\log(K/S_0)^2} 
\Sigma_{\rm LN}(K/S_0) \,.
\end{equation}
For $K/S_0 \sim 1$ around the ATM point, we have the expansion
\begin{equation}
\Sigma''_{\rm LN}(K/S_0) := \frac{\sigma}{\sqrt3} \left( 
-\frac{23}{1050}  + \frac{3}{1750} x + O\left(x^2\right) \right)\,.
\end{equation}

The first line in the result (\ref{GammaC}) is identical to the well-known 
result for Gamma
in the Black-Scholes model. The second and third terms depend on the
slope and curvature of the equivalent log-normal volatility of the 
Asian option $\Sigma_{\rm LN}(K/S_0)$. These corrections are non-vanishing
in the $T\to 0$ limit.

\textbf{Vega}

The \textit{Vega} of the approximated Asian option is defined as 
\begin{equation}
\mathcal{V}_{\overline{C}}  := 
\left( \frac{\partial \overline{C}}{\partial \sigma} \right)\,,
\end{equation}
and is given by
\begin{equation}
\mathcal{V}_{\overline{C}}  = e^{-rT} A(T) \Sigma_{\rm LN}(K/S_0) \sqrt{T} 
\frac{1}{\sqrt{2\pi} \sigma} e^{-\frac12 d_1^2} \,.
\end{equation}
This is the same for the call and put Asian options, as can be seen from 
the put-call parity relation (\ref{PutCall}). 

\textbf{Rho}

The \textit{Rho} of the approximated Asian call option is defined as 
\begin{equation}
\text{Rho}_{\overline{C}} = \left( \frac{\partial \overline{C}}{\partial r} \right)\,.
\end{equation}
Noting that $\Sigma_{\rm LN}(K/S_0)$ does not depend on $r$, we get
that the Rho of an Asian call option is given by
\begin{equation}\label{RhoC}
\text{Rho}_{\overline{C}} = e^{-rT} T [- A(T) N(d_1) R((r-q)T) + K N(d_2) ]\,,
\end{equation}
where the function $R(x)$ is 
\begin{equation}\label{Rdef}
R(x) = 1 + \frac{1}{x} - \frac{1}{1-e^{-x}} = \frac12 - \frac{1}{12}x + O(x^2)\,.
\end{equation}
The Rho of the approximated Asian put option can be obtained using the put-call
parity (\ref{PutCall})
\begin{equation}
\text{Rho}_{\overline{C}} - \text{Rho}_{\overline{P}} = T e^{-rT} K + e^{-rT} S_0 T f((r-q)T)\,,
\end{equation}
with 
\begin{equation}
f(x):=\frac{1-e^x}{x^2} + \frac{1}{x}\,.
\end{equation}

\subsection{The Greeks of floating-strike Asian options}

The Greeks of the floating-strike Asian options can be obtained from those
of the fixed-strike Asian options using the fixed-floating strike symmetry
of \cite{HW}. They are given below in 
(\ref{DeltaCFloat}), (\ref{DeltaPFloat}),
(\ref{GammaFloat}) and (\ref{VegaFloat}). Combining them with the short
maturity approximations derived above, they can be used
to obtain analytical approximations also for the sensitivities of the
floating-strike Asian options.

\subsection{Numerical tests}

We compare in this Section the analytical approximations for the Greeks of the 
Asian options
against numerical results obtained using alternative implementations from
\cite{3,4,BP}.

\textbf{Scenario 1.}  One first scenario was considered by \cite{3},
see Table 3 in this paper. These authors use the following scenario
\begin{equation}\label{scenario1}
r = 0.1\,, K=100\,, \sigma=0.25\,, T = 0.5 \,.
\end{equation}

In Table~\ref{tab:Delta} we show the results for the sensitivities of an Asian
call option obtained in \cite{3} using three different
methods: finite
differences (FD), pathwise (PW) and likelihood method (Lkhd). For each 
sensitivity we show the three results from Table 3 of \cite{3}, and
the short maturity approximations given above.

For all cases when the three results (FD, PW, Lkhd) are close, the 
short maturity approximation is in reasonably good agreement with the
alternative methods (except for Gamma under the FD method with $S_0=90$
which differs by an order of magnitude; this is most likely a typo).
In some cases FD and PW are close, but differ significantly from Lkhd. In
these cases the short maturity approximation is closer to FD and PW.


\begin{table}[htbp]
  \centering
  \caption{Numerical results for the Greeks of a call Asian option under
the scenario (\ref{scenario1}) and the shown values of $S_0$, from Table 3
of \cite{3}.
The first three rows give the finite
difference (FD), pathwise (PW) and likelihood (Lkhd) methods of \cite{3}. 
The last rows give the short maturity approximations
for the sensitivities of this paper.}
    \begin{tabular}{|l|ccc|}
    \hline
Delta & $S_0=90$ & $S_0=100$ & $S_0=105$ \\
    \hline\hline
FD   & $0.2257\pm 0.001$ & $0.5991 \pm 0.0027$ & $0.7601 \pm 0.0034$ \\
PW   & $0.2255\pm 0.001$ & $0.5990 \pm 0.0027$ & $0.7607 \pm 0.0034$ \\
Lkhd & $0.2111\pm 0.001$ & $0.5880 \pm 0.0027$ & $0.7615 \pm 0.0035$ \\
    \hline
$\Delta_{\overline C}$        & 0.2250 & 0.5978 & 0.7596 \\
\hline
Gamma & $S_0=90$ & $S_0=100$ & $S_0=105$ \\
    \hline\hline
FD   & $0.31134\pm 0.00014$ & $0.03730 \pm 0.00017$ & $0.02712 \pm 0.00012$ \\
PW   & $0.03230\pm 0.00015$ & $0.03684 \pm 0.00017$ & $0.02566 \pm 0.00012$ \\
Lkhd & $0.02282\pm 0.00045$ & $0.02340 \pm 0.00106$ & $0.01887 \pm 0.00130$ \\
    \hline
$\Gamma_{\overline C}$ & 0.03259 & 0.03694  & 0.02730 \\
\hline
Vega & $S_0=90$ & $S_0=100$ & $S_0=105$ \\
    \hline\hline
FD   & $11.07980\pm 0.04956$ & $15.13080 \pm 0.06766$ & $12.17077 \pm 0.05444$ \\
PW   & $11.07959\pm 0.04955$ & $15.13088 \pm 0.06766$ & $12.16922 \pm 0.05444$ \\
Lkhd & $ 8.92731\pm 0.04514$ & $12.82907 \pm 0.08673$ & $10.65001 \pm 0.01039$ \\
    \hline
$\mathcal{V}_{\overline C}$ & 11.0942 & 15.2009 & 12.2608 \\
\hline
Rho & $S_0=90$ & $S_0=100$ & $S_0=105$ \\
    \hline\hline
FD   & $4.66913\pm 0.04188$ & $12.66914 \pm 0.04930$ & $15.99043 \pm 0.04213$ \\
PW   & $4.28253\pm 0.01916$ & $12.01245 \pm 0.05371$ & $15.30865 \pm 0.06845$ \\
Lkhd & $4.64600\pm 0.02078$ & $12.57021 \pm 0.05622$ & $15.81224 \pm 0.07072$ \\
    \hline
$Rho_{\overline C}$ & 4.58395 & 12.5415 & 15.8546 \\
\hline
    \end{tabular}%
  \label{tab:Delta}%
\end{table}%


\textbf{Scenario 2.} We compare next with the benchmark results of \cite{4}.
This paper considers a set of ATM scenarios
\begin{equation}\label{scenario2ATM}
r = 0.05\,, S_0 = K=100\,, \sigma=\{ 0.1, 0.3, 0.5 \}\,, 
\end{equation}
and a set of OTM scenarios
\begin{equation}\label{scenario2OTM}
r = 0.05\,, S_0 =100\,, K=105\,, \sigma= \{ 0.1, 0.3, 0.5 \}\,.
\end{equation}
For each case the option maturity is chosen as
$T = \{1.0, 0.5, 0.25, 1/12\} $. The $T=1$ case is also considered
in \cite{BP}.

The results for the Asian option prices corresponding to the two sets
of scenarios are shown in Table~\ref{tab:ZhenyuATM} (ATM scenarios)
and Table~\ref{tab:ZhenyuOTM} (OTM scenarios). For each case we compare
with the Continuous Time Markov Chain (CMTC) method
of \cite{4} (Table 1 (ATM) and Table 3 (OTM) in \cite{4}).
Additional benchmark values obtained using a
Monte Carlo approach are given in \cite{4,BP}, but are not reproduced here for
reasons of space economy.

The short maturity approximation for the Asian option prices are in reasonably
good agreement with alternative simulation methods for $\sigma<0.5$. The
differences increase with the volatility and maturity, as expected.

We also show the short maturity Asian sensitivities
in Table~\ref{tab:ZhenyuATMGreeks} (ATM scenarios)
and Table~\ref{tab:ZhenyuOTMGreeks} (OTM scenarios).
These are compared against the results from the PDE 
method due to 
\cite{Vecer1,Vecer2} and the Continuous Time Markov Chain (CMTC) method
of \cite{4} in Table 2 (ATM) and Table 4 (OTM) of \cite{4}. 
The agreement is again reasonably good, and becomes worse for larger 
volatility $\sigma$ and maturity $T$.

In conclusion, the short maturity approximation for the Greeks of Asian
options in the Black-Scholes model presented here is in reasonably
good agreement with alternative simulation methods, for sufficiently small
variance $\sigma^2 T$ and small interest rates $rT$.  The interest rates
effects can be taken also into account using the approach
presented in \cite{PZAsian2017}. Under this approach, the approximation
(\ref{CAsian}) is replaced with Eq.~(120), and the equivalent
log-normal volatility $\Sigma_{\rm LN}(K/S_0)$ is replaced with the result
of Proposition 19 in this reference.
Numerical tests presented in Section 6 of
\cite{PZAsian2017} demonstrate good agreement for Asian option prices with the
precise method of \cite{Linetsky}. A similar approximation can be used also to
obtain the Greeks for Asian options.

\begin{table}[htbp]
  \centering
  \caption{Numerical results for Asian option prices under the 
ATM scenario from Table 1 of \cite{4}. The first two columns show 
the results from the Continuous Time Markov Chain (CTMC) method of \cite{4}, 
and from the PDE method of \cite{Vecer1,Vecer2}. The
last column shows the results from the short maturity asymptotics
of Eq.~(\ref{CAsian}).}
    \begin{tabular}{|lc|ccc|}
    \hline
$\sigma$ & $T$ & CTMC & Vecer & $\overline{C}$ \\
    \hline\hline
0.1 & 1   & 3.64116 & 3.64118 & 3.63401 \\
0.1 & 0.5 & 2.31599 & 2.29024 & 2.28768 \\
0.1 & 0.25 & 1.47979 & 1.47814 & 1.47747 \\
0.1 & 1/12 & 0.76807 & 0.77106 & 0.77187 \\
0.1 & 1/252 & 0.14363 & 0.14998 & 0.15009 \\
\hline
0.3 & 1   & 7.94659 & 7.94550 & 7.92675 \\
0.3 & 0.5 & 5.44560 & 5.43886 & 5.43212 \\
0.3 & 0.25 & 3.75502 & 3.74268 & 3.74049 \\
0.3 & 1/12 & 2.09084 & 2.09393 & 2.09389 \\
0.3 & 1/252 & 0.43554 & 0.43871 & 0.44019 \\
\hline
0.5 & 1   & 12.32052 & 12.32063 & 12.30620 \\
0.5 & 0.5 &  8.61829 &  8.61333 &  8.60816 \\
0.5 & 0.25 & 6.02426 &  6.01662 &  6.01486 \\
0.5 & 1/12 & 3.41636 &  3.41826 &  3.41810 \\
0.5 & 1/252 & 0.72532 & 0.72944 & 0.73032 \\
\hline
    \end{tabular}%
  \label{tab:ZhenyuATM}%
\end{table}%


\begin{table}[htbp]
  \centering
  \caption{Numerical results for the Greeks of the 
Asian option prices for the ATM scenarios (\ref{scenario2OTM}) of \cite{4}. 
We compare the short maturity asymptotics results
against the alternative methods (CTMC \citep{4} and \cite{Vecer1,Vecer2}) 
in Table 2 of \cite{4}. }
    \begin{tabular}{|c|ccc|ccc|ccc|}
    \hline
\multicolumn{1}{|c|}{} & \multicolumn{3}{|c|}{Delta} 
                       & \multicolumn{3}{|c|}{Gamma} 
                       & \multicolumn{3}{|c|}{Vega} \\
\hline
$T$            & CTMC & Vecer & $\Delta_{\overline{C}}$ 
               & CTMC & Vecer & $\Gamma_{\overline{C}}$ 
               & CTMC & Vecer & $\mathcal{V}_{\overline{C}}$ \\
    \hline\hline
\multicolumn{10}{|c|}{$\sigma=0.1$} \\
\hline
1   & 0.6646 & 0.6591 & 0.6599 
    & 0.0635 & 0.0607 & 0.0610      
    & 19.9635 & 20.2806 & 20.1767 \\
0.5 & 0.6146 & 0.6255 & 0.6189 
    & 0.0884 & 0.0833 & 0.0920      
    & 14.9495 & 15.2858 & 15.2463 \\
0.25& 0.5788 & 0.5971 & 0.5866 
    & 0.1071 & 0.1124 & 0.1342      
    & 11.0037 & 11.1182 & 11.1430 \\
1/12& 0.5373 & 0.5732 & 0.5512 
    & 0.2375 & 0.2370 & 0.2372      
    & 6.2215 & 6.5543 & 6.5764 \\
\hline
\multicolumn{10}{|c|}{$\sigma=0.3$} \\
\hline
1   & 0.5723 & 0.5722 & 0.5704 
    & 0.0121 & 0.0155 & 0.0223      
    & 21.8410 & 21.8961 & 21.8724 \\
0.5 & 0.5552 & 0.5555 & 0.5532 
    & 0.0327 & 0.0319 & 0.0322      
    & 15.8379 & 15.8794 & 15.8712 \\
0.25& 0.5402 & 0.5435 & 0.5392 
    & 0.042 & 0.0427  & 0.0459      
    & 11.2979 & 11.3607 & 11.3686 \\
1/12& 0.527 & 0.532 & 0.5235 
    & 0.0743 & 0.0794 & 0.0799      
    & 6.4603 & 6.6337 & 6.6205 \\
\hline
\multicolumn{10}{|c|}{$\sigma=0.5$} \\
\hline
1   & 0.5672 & 0.5668 & 0.5660 
    & 0.0065 & 0.0072 & 0.0135      
    & 21.7296 & 21.7987 & 21.8736 \\
0.5 & 0.5513 & 0.5514 & 0.5500 
    & 0.0190 & 0.0145 & 0.0194      
    & 15.8230 & 15.8096 & 15.8715 \\
0.25& 0.5382 & 0.5392 & 0.5369 
    & 0.0219 & 0.0216  & 0.0277      
    & 11.3402 & 11.3593 & 11.3687 \\
1/12& 0.5238 & 0.5266 & 0.5222 
    & 0.0435 & 0.0457 & 0.0481     
    & 6.5693 & 6.6187 & 6.6205 \\
\hline
    \end{tabular}%
  \label{tab:ZhenyuATMGreeks}%
\end{table}%

\begin{table}[htbp]
  \centering
  \caption{Numerical results for Asian option prices under the 
OTM scenarios (\ref{scenario2OTM}) of \cite{4}, comparing with the 
Continuous Time Markov
Chain (CTMC) method of \cite{4} and with the PDE method of Vecer,
see Table 3 in \cite{4}. The
last column shows the results from the short maturity asymptotics
of Eq.~(\ref{CAsian}).}
    \begin{tabular}{|lc|ccc|}
    \hline
$\sigma$ & $T$ & CTMC & Vecer & $\overline{C}$ \\
    \hline\hline
0.1 & 1   & 1.31217 & 1.31063 & 1.30304 \\
0.1 & 0.5 & 0.44555 & 0.42856 & 0.42658 \\
0.1 & 0.25 & 0.09059 & 0.09362 & 0.09314 \\
0.1 & 1/12 & 0.00145 & 0.00144 & 0.00135 \\
\hline
0.3 & 1   & 5.74781 & 5.758607 & 5.73937 \\
0.3 & 0.5 & 3.33363 & 3.33079 & 3.32412 \\
0.3 & 0.25 & 1.78915 & 1.77951 & 1.77730 \\
0.3 & 1/12 & 0.48821 & 0.48819 & 0.48797 \\
\hline
0.5 & 1   & 10.26652 & 10.26642 & 10.2519 \\
0.5 & 0.5 &  6.53065 &  6.52547 &  6.52029 \\
0.5 & 0.25 & 3.97055 &  3.96283 &  3.96104 \\
0.5 & 1/12 & 1.55832 &  1.54391 &  1.54372 \\
\hline
    \end{tabular}%
  \label{tab:ZhenyuOTM}%
\end{table}%


\begin{table}[htbp]
  \centering
  \caption{Numerical results for the Greeks of the 
Asian option prices for the OTM scenarios (\ref{scenario2OTM}) of \cite{4}. 
We compare the short maturity asymptotics results
against the alternative methods (CTMC and Vecer) in Table 4 of \cite{4}. }
    \begin{tabular}{|c|ccc|ccc|ccc|}
    \hline
\multicolumn{1}{|c|}{} & \multicolumn{3}{|c|}{Delta} 
                       & \multicolumn{3}{|c|}{Gamma} 
                       & \multicolumn{3}{|c|}{Vega} \\
\hline
$T$ & CTMC & Vecer & $\Delta_{\overline{C}}$ 
    & CTMC & Vecer & $\Gamma_{\overline{C}}$ 
    & CTMC & Vecer & $\mathcal{V}_{\overline{C}}$ \\
    \hline\hline
\multicolumn{10}{|c|}{$\sigma=0.1$} \\
\hline
1   & 0.3457 & 0.3444 & 0.3415 
    & 0.0613 & 0.0645 & 0.0629      
    & 20.5678 & 21.1117 & 21.0090 \\
0.5 & 0.1992 & 0.1929 & 0.1904
    & 0.0643 & 0.0645 & 0.0665      
    & 10.7848 & 11.2099 & 11.1330 \\
0.25& 0.0890 & 0.0734 & 0.0724 
    & 0.0522 & 0.0512 & 0.0478      
    & 4.0080 & 4.0693 & 4.0090 \\
1/12& 0.0034 & 0.0031 & 0.0027 
    & 0.0522 & 0.0512 & 0.0050      
    & 0.1432 & 0.1418 & 0.1397 \\
\hline
\multicolumn{10}{|c|}{$\sigma=0.3$} \\
\hline
1   & 0.4631 & 0.4631 & 0.4619 
    & 0.0218 & 0.0220 & 0.0228      
    & 22.4796 & 22.5724 & 22.5485 \\
0.5 & 0.3998 & 0.3993 & 0.3981
    & 0.0309 & 0.0310 & 0.0316      
    & 15.7093 & 15.7346 & 15.7296 \\
0.25& 0.3252 & 0.3238 & 0.3229 
    & 0.0405 & 0.0412  & 0.0416      
    & 10.3772 & 10.4244 & 10.4150 \\
1/12& 0.1871 & 0.1848 & 0.1812 
    & 0.0535 & 0.0537 & 0.0529      
    & 4.3964 & 4.4357 & 4.4290 \\
\hline
\multicolumn{10}{|c|}{$\sigma=0.5$} \\
\hline
1   & 0.5020 & 0.5019 & 0.5011 
    & 0.0130 & 0.0132 & 0.0138      
    & 22.0044 & 22.4501 & 22.5303 \\
0.5 & 0.4577 & 0.4573 & 0.4565 
    & 0.0189 & 0.0190 & 0.0195      
    & 16.0843 & 16.0843 & 16.1184 \\
0.25& 0.4060 & 0.4573 & 0.4045 
    & 0.0267 & 0.0264  & 0.0271      
    & 11.1946 & 11.2161 & 11.2235 \\
1/12& 0.3018 & 0.3023 & 0.2994 
    & 0.0412 & 0.0416 & 0.0420     
    & 5.7826 & 5.8367 & 5.8366 \\
\hline
    \end{tabular}%
  \label{tab:ZhenyuOTMGreeks}%
\end{table}%

\section{Qualitative Behavior of Asian Greeks}\label{Sec:Qualitative}

Qualitative behavior of the sensitivities for Asian options 
has been studied rigorously for certain Greeks in the literature.
Let us recall that in the Black-Scholes model, the fixed-strike Asian call and put option 
prices are given by:
\begin{align}
&C(K,T,\sigma,r,q)=e^{-rT}\mathbb{E}
\left[\left(\frac{1}{T}\int_{0}^{T}S_{0}e^{(r-q-\frac{1}{2}\sigma^{2})t+\sigma W_{t}}dt-K\right)^{+}\right],\label{CBS}
\\
&P(K,T,\sigma,r,q)=e^{-rT}\mathbb{E}\left[\left(K
-\frac{1}{T}\int_{0}^{T}S_{0}e^{(r-q-\frac{1}{2}\sigma^{2})t+\sigma W_{t}}dt\right)^{+}\right].\label{PBS}
\end{align}
The Greeks are defined as the partial derivatives of $C$ and $P$ in \eqref{CBS} and \eqref{PBS}
with respect to the model parameters.

In the Black-Scholes model, the floating-strike Asian call and put option 
prices are given by:
\begin{align}
&C_f(\kappa,T,\sigma,r,q):=\mathbb{E}\left[\left(\kappa S_{0}e^{-qT-\frac{1}{2}\sigma^{2}T+\sigma W_{T}}
-\frac{1}{T}\int_{0}^{T}S_{0}e^{r(t-T)-qt-\frac{1}{2}\sigma^{2}t+\sigma W_{t}}dt\right)^{+}\right],\label{CfBS}
\\
&P_f(\kappa,T,\sigma,r,q):=e^{-rT}\mathbb{E}\left[\left(\frac{1}{T}\int_{0}^{T}S_{0}e^{r(t-T)-qt-\frac{1}{2}\sigma^{2}t+\sigma W_{t}}dt
-\kappa S_{0}e^{-qT-\frac{1}{2}\sigma^{2}T+\sigma W_{T}}\right)^{+}\right].\label{PfBS}
\end{align}
The Greeks are defined as the partial derivatives of $C_f$ and $P_f$ in \eqref{CfBS} and \eqref{PfBS}
with respect to the model parameters.

\cite{HW} proved an equivalence relation between the 
fixed- and floating-striked Asian options
under the Black-Scholes model:
\begin{align}
&C_{f}(\kappa,T,\sigma,r,q)=P(\kappa S_{0},T,\sigma,q,r),
\\
&P_{f}(\kappa,T,\sigma,r,q)=C(\kappa S_{0},T,\sigma,q,r).
\end{align}
We will use these relations in the next section to express 
sensitivities of the
floating-strike Asian options in terms of those of the fixed-strike
Asian options.

\subsection{Delta, Gamma and Vega of the Asian options}

The signs of Delta and Gamma for Asian options can be easily determined. 
For fixed-strike Asian options, from Eq. \eqref{CBS} and Eq. \eqref{PBS}, 
it is clear that the Asian call (put) option price increases (decreases) 
convexly with respect to spot price $S_{0}$.
Therefore Delta is positive for the call option and negative for the put option,
and Gamma is always positive.

For floating-strike Asian options, it is clear that both $C_f$ and $P_f$ are 
linear in $S_{0}$ and we have
\begin{align}\label{DeltaCFloat}
&\Delta_{C,f}=\mathbb{E}\left[\left(\kappa e^{-qT-\frac{1}{2}\sigma^{2}T+\sigma W_{T}}
-\frac{1}{T}\int_{0}^{T}e^{r(t-T)-qt-\frac{1}{2}\sigma^{2}t+\sigma W_{t}}dt\right)^{+}\right]
\\
&\qquad= \frac{1}{S_0} C_f(\kappa, T, \sigma, r, q) = 
\frac{1}{S_0} P(\kappa S_0, T, \sigma, q, r)\,, \nonumber \\
\label{DeltaPFloat}
&\Delta_{P,f}=e^{-rT}\mathbb{E}\left[\left(\frac{1}{T}\int_{0}^{T}e^{r(t-T)-qt-\frac{1}{2}\sigma^{2}t+\sigma W_{t}}dt
-\kappa e^{-qT-\frac{1}{2}\sigma^{2}T+\sigma W_{T}}\right)^{+}\right],\\
&\qquad = \frac{1}{S_0} P_f(\kappa, T, \sigma, r, q) =
\frac{1}{S_0} C(\kappa S_0, T, \sigma, q, r) \nonumber\,.
\end{align}
Thus the Delta of floating-strike Asian options are positive.
Their Gamma is zero:
\begin{equation}\label{GammaFloat}
\Gamma_{C,f}=\Gamma_{P,f}=0.
\end{equation}

The sign for the Vega is surprisingly difficult to analyze. 
\cite{Carr} showed rigorously that 
the Vega is always positive for the Black-Scholes model for the fixed-strike Asian call 
options, see \cite{BY} for an alternative proof. 
By the put-call parity, the Vega is positive for the fixed-strike Asian put 
option.
By the equivalence relation between fixed-strike and floating-strike Asian 
options \citep{HW}, one has the relation between the Vega sensitivities of
these two types of Asian options
\begin{equation}\label{VegaFloat}
\mathcal{V}_{C,f}(\kappa,T,\sigma,r,q) = \mathcal{V}_C(\kappa S_0,T,\sigma,q,r)\,.
\end{equation}
We conclude that the Vega is also positive for the floating-strike Asian 
options.

We summarize what is known about the signs of Delta, Gamma, Vega in 
Table \ref{KnownTable}.

\begin{table}[htb]
\centering 
\caption{Summary of the signs of Delta, Gamma, Vega for Asian options
in the Black-Scholes model. }
\begin{tabular}{|c||c|c|c|} 
\hline 
Option & sign of Delta & sign of Gamma & sign of Vega
\\
\hline
\hline
fixed-strike Asian call & $+$ & $+$ & $+$
\\
\hline
fixed-strike Asian put & $-$ & $+$ & $+$ 
\\
\hline
floating-strike Asian call & $+$ & $0$ & $+$
\\
\hline
floating-strike Asian put & $+$ & $0$ & $+$
\\
\hline 
\end{tabular}
\label{KnownTable} 
\end{table}

To our knowledge, there is no rigorous study in the literature of 
the signs of the Greek letters Rho and Psi, the sensitivities
with respect to the risk-free rate $r$ and the dividend yield $q$ 
for the Asian options for both fixed- and floating-strikes. 
This will be our focus in Section \ref{ctsFixed} and Section \ref{ctsFloating}. 

We prove rigorously that in the Black-Scholes model 
Rho for the fixed-strike Asian call option can be both positive and 
negative, and Rho for the fixed-strike Asian put option is always negative.
Rho is always positive for the floating-strike Asian call option,
and for the floating-strike Asian put option. Similar results are obtained 
for Psi, the sensitivity with respect to the dividend yield.
We summarize our results in Table \ref{SummaryTable}.
The rigorous derivations will be given in Section \ref{ctsFixed} and Section \ref{ctsFloating}.

\begin{table}[htb]
\centering 
\caption{Summary of the sign of Rho and Psi for Asian options
in the Black-Scholes model. }
\begin{tabular}{|c||c|c|} 
\hline 
Option & sign of Rho & sign of Psi \\
\hline
\hline
fixed-strike Asian call & $+/-$ & $-$
\\
\hline
fixed-strike Asian put & $-$ & $+$ 
\\
\hline
floating-strike Asian call & $+$ & $-$
\\
\hline
floating-strike Asian put & $-$ & $+/-$ 
\\
\hline 
\end{tabular}
\label{SummaryTable} 
\end{table}

\subsection{Rho and Psi for Fixed-strike Asian options}\label{ctsFixed}

We have the following results for Rho and Psi of fixed-strike Asian options in the
Black-Scholes model. 

\begin{proposition}\label{prop:1}
(i) The Rho for the fixed-strike Asian put options is always negative.

(ii) The Rho for the fixed-strike Asian call options is sometimes negative and sometimes positive.

(iii) The Psi for the fixed-strike Asian put options is always positive.

(iv) The Psi for the fixed-strike Asian call options is always negative.
\end{proposition}

\begin{proof}
For the Asian put options,
\begin{equation}
P(K,T,\sigma,r,q)=e^{-rT}\mathbb{E}\left[\left(K-\frac{1}{T}\int_{0}^{T}e^{(r-q-\frac{1}{2}\sigma^{2})t+\sigma W_{t}}dt\right)^{+}\right].
\end{equation}
Therefore, it is clear that 
$\left(K-\frac{1}{T}\int_{0}^{T}e^{(r-q-\frac{1}{2}\sigma^{2})t+\sigma W_{t}}dt\right)^{+}$ is decreasing in $r$.
Trivially, the discount factor $e^{-rT}$ is decreasing in $r$ and increasing
in $q$. Hence, $P(r,q,\sigma,K,T)$ is decreasing in $r$ and increasing in $q$.
We conclude that Rho for the Asian put options is always negative, and Psi is 
always positive. This proves (i) and (iii).

A similar argument gives that the Psi of the Asian call option is always
negative, which proves (iv).
However the analysis for Rho of the Asian call option is much more complicated.
By put-call parity,
\begin{equation}
F(r)=C(K,T,\sigma,r,q)-P(K,T,\sigma,r,q)
=e^{-rT}\mathbb{E}\left[\frac{1}{T}\int_{0}^{T}S_{t}dt-K\right]
=e^{-rT}S_{0}\frac{e^{(r-q)T}-1}{(r-q)T}-e^{-rT}K.
\end{equation}
Let us assume that
\begin{equation}
F(0)=S_{0}\frac{1-e^{-qT}}{qT}-K>0.
\end{equation}
Note that $\lim_{r\rightarrow\infty}F(r)=0$. 
Therefore $F(r)$ must be decreasing in $r$ on some interval. 
Since $\partial C/\partial r=\partial F/\partial r+\partial P/\partial r$
and $\partial P/\partial r$ is always negative, we conclude
that Rho is negative for the Asian call options for some values of the
model and option parameters.

We will show that Rho can also become positive. For this it is convenient to 
write this sensitivity in a more explicit form as (we denote here and below
the instrument type by a subscript $C,P$ on Rho):
\begin{align}\label{5}
\text{Rho}_C &= e^{-rT} \left[ -T(A_T - K) 1_{A_T\geq K}] + 
   \partial_r \mathbb{E}[1_{A_T\geq K} (A_T - K)] \right] \\
&= e^{-rT} \left( -T \mathbb{E}[(A_T - K) 1_{A_T>K}] + 
                      \mathbb{E}[B_T 1_{A_T>K}]\right) \nonumber \\
&= e^{-rT} \left( T K \mathbb{P}(A_T \geq K)  - 
   \mathbb{E}[(T A_T - B_T) 1_{A_T \geq K}]\right)\,, \nonumber
\end{align}
with
\begin{equation}
A_T:=\frac{1}{T}\int_{0}^{T}S_{t}dt\,,
\qquad
B_T := \frac{1}{T} \int_0^T t S_t dt \,.
\end{equation}

This relation is the continuous time analog of the pathwise estimator for Rho 
obtained in Proposition 5 of \cite{BrGl} for an Asian option with discrete-time
averaging, see Eq. (45) in this paper. 
We fixed a small typo in their Eq. (45).

The second term in (\ref{5}) is positive, as one can see by writing
it as
\begin{equation}\label{6}
T A_T - B_T = \frac{1}{T} \int_0^T (T-t) S_t dt> 0 \,.
\end{equation}
The sign of Rho$_C$ is determined by the relative size of the two terms. 

The expression \eqref{5} can help us to show that $\text{Rho}_{C}$ can take both negative and
positive values. First, observe that
\begin{equation}
\lim_{r\rightarrow\infty}\mathbb{P}(A_{T}>K)=1,
\qquad
\lim_{r\rightarrow\infty}\mathbb{E}[(TA_{T}-B_{T})1_{A_{T}>K}]=\infty.
\end{equation}
Therefore, we conclude that for any fixed $K,S_{0},q,T,\sigma$, $\text{Rho}_C$ is 
always negative for sufficiently large $r$.

Next, assume $r=q=\sigma=0$, and let $S_{0}>K$. 
Then, we have $A_{T}=S_{0}$ and 
\begin{equation}
TA_{T}-B_{T}=\frac{1}{T}\int_{0}^{T}(T-t)S_{0}dt=S_{0}\frac{T}{2}\,.
\end{equation}
Thus we have
\begin{equation}
\text{Rho}_C=TK-S_{0}\frac{T}{2}>0,
\end{equation}
if we have $2K>S_{0}>K$. Therefore, for some parameters, the $\text{Rho}_{C}$ for the 
Asian call option is positive. 

The restriction above, $\sigma=0$, is a little bit too strong. 
Indeed, we can relax this assumption and argue as follows.
Rho$_C$ is bounded from below as
\begin{equation}\label{RhoLowerBound}
\text{Rho}_C \geq e^{-rT} 
 \left( T K \mathbb{P}(A_T>K)  - \mathbb{E}[T A_T - B_T]\right).
\end{equation}
Let us study how the right hand side of the equation \eqref{RhoLowerBound} depends on on $K$. 
The first term is positive, 
and the second term contributes with a negative sign, see (\ref{6}), and 
is $K$ independent. Its expectation is finite
\begin{equation}
\mathbb{E}[T A_T - B_T]= \frac{S_0}{T} 
\left( \frac{1}{(r-q)^2} e^{(r-q)T} - \frac{1}{(r-q)^2} (1 + (r-q)T) \right)\,.
\end{equation}
for $r-q\neq 0$ and $\mathbb{E}[T A_T - B_T]=TS_{0}/2$ for $r-q=0$.
We will show next that the maximum of the first term as $K$ increases
can sometimes become larger than the second term, and thus Rho$_C$ can become positive in 
a range of strikes. 

Let us study whether $T K \mathbb{P}(A_T>K)$ (from the first term on the right hand side in (\ref{RhoLowerBound})) is
ever larger than the second one. For the simple case $r-q=0$, it is equivalent to whether
\begin{equation}
K \mathbb{P}(A_T > K) > \frac12 S_0 ,
\end{equation}
for some value of $K$. Since the arithmetic mean is always larger than the geometric mean, we have
\begin{equation}
\mathbb{P}(A_T > K ) \geq \mathbb{P}(G_T > K),
\end{equation}
where
\begin{equation}
G_{T}:=\exp\left\{\frac{1}{T}\int_0^T \log S_t dt\right\}.
\end{equation}
The distribution of the geometric mean $G_T$ can be found exactly and $G_{T}$ can be represented as
\begin{equation}\label{Gresult}
G_T = S_0 \exp\left( \frac{1}{\sqrt3} \sigma \sqrt{T} Z - \frac14 \sigma^2 T + \frac12 (r-q) T \right)\,,
\end{equation}
where $Z \sim N(0,1)$ and the equality \eqref{Gresult} holds in distribution.
This gives
\begin{equation}
\mathbb{P}(G_T > K) = \int_{z_0(K)}^\infty \frac{e^{-\frac12 z^2}}{\sqrt{2\pi}}
dz = 1 - N(z_0(K)),
\end{equation}
with
\begin{equation}\label{Z0res}
z_0(K) := \frac{\sqrt3}{\sigma\sqrt{T}}
\left( \log \frac{K}{S_0} + \frac14 \sigma^2 T - \frac12 (r-q) T \right).
\end{equation}
This can be used in (\ref{RhoLowerBound}) to study numerically whether the 
lower bound ever becomes positive. Numerical study (see Remark \ref{NumRemark}) shows that this does happen
indeed, for certain values of the parameters so that
\begin{equation}
\sup_{K>0}K \mathbb{P}(A_T > K)\geq\sup_{K>0}K\mathbb{P}(G_{T}>K) > \frac12 S_0 .
\end{equation}
This concludes the proof of (ii). An alternative analytical approach that does not use any numerical computations
will be given in the Appendix~\ref{ctsAlt}.
\end{proof}

\begin{remark}\label{NumRemark}
As a numerical example, take $T=1$, $\sigma = 0.2$ and $r=q=0$. 
For this case the supremum of $K\mathbb{P}(G_T > K)$ over $K$ is $0.78 S_0$ and 
is reached at $K=0.82 S_0$. This is larger than $0.5$ so Rho$_C$ must become 
positive in a neighborhood of this value of the strike. Keeping $T=1$ and 
varying $\sigma$, this remains true for sufficiently small volatility 
$\sigma < 0.714$. If the volatility $\sigma$ is above $0.714$, then the 
supremum of $K\mathbb{P}(G_T>K)$ is below $0.5$.
\end{remark}

\begin{remark}
An alternative proof of positive $\text{Rho}_C$ in certain regions of the parameter
space can be given using the alternative formula for $\text{Rho}_C$ from the Malliavin calculus
with the analysis that relies on the large deviations theory \citep{LD,VarLD}. 
We sketch the main steps in the Appendix.
\end{remark}

\subsection{Zeros of $\text{Rho}_{\overline C}$}

In this section we study the position of the zeros of the short maturity
approximation $\text{Rho}_{\overline C}$ given in (\ref{RhoC}). 

\begin{proposition}
The set of the points in the $(K/S_0, (r-q)T)$ plane where 
$Rho_{\overline C}=0$ can be parameterized as
\begin{equation}\label{ZeroRho}
\frac{K}{S_0} = \kappa(v(K/S_0), R((r-q)T)) \frac{e^{(r-q)T}-1}{(r-q)T}\,,
\end{equation}
where $v(K/S_0) = \Sigma_{\rm LN}(K/S_0)\sqrt{T}$ and $R(x)$ is given in 
(\ref{Rdef}). The function $\kappa(v,R)$ is defined as the solution of the
equation
\begin{equation}\label{kappadef}
\frac{\kappa}{R} N\left(-\frac{1}{v} \log \kappa - \frac12 v\right) = 
N\left(- \frac{1}{v}\log \kappa + \frac12 v \right)\,,\quad R>0\,.
\end{equation}
\end{proposition}

\begin{proof}
From (\ref{RhoC}) we have that the equation $Rho_{\overline C}=0$ depends
only on $K/A(T)$ and $(r-q)T$ and can be put into the form (\ref{kappadef})
with $K/S_0 \to \kappa$.
For any $R>0$ this equation has a unique solution for $\kappa$, which defines
the function $\kappa(v,R)$.
\end{proof}

Numerical evaluation shows that $\kappa(v,R)$ is an increasing function of 
$v \in (0,\infty)$.
This reaches the value $\lim_{v\to 0} \kappa(v,R) = R$ as $v\to 0$. 
The plot of this function for $R=1/2$ is given in Figure~\ref{Fig:1} (left).

\begin{figure}
    \centering
   \includegraphics[width=2.5in]{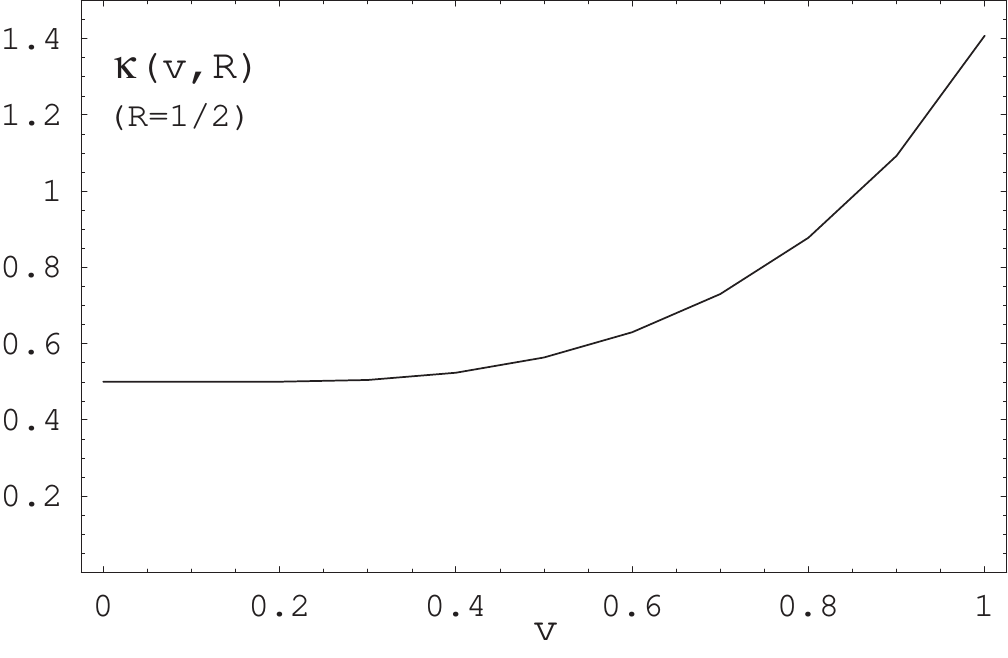}
   \includegraphics[width=2.6in]{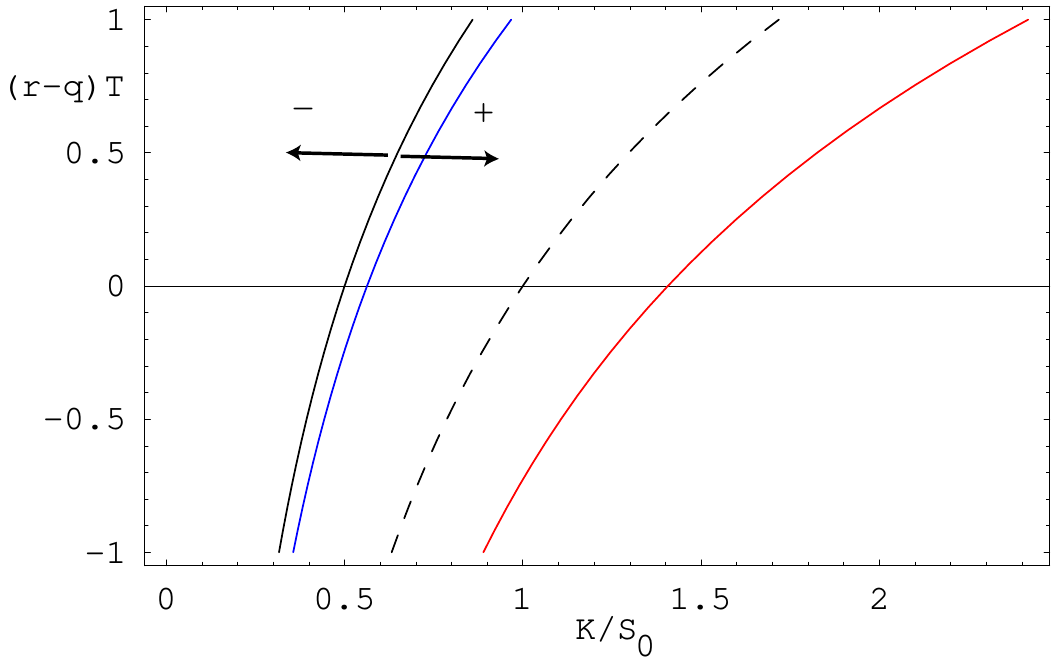}
    \caption{
Left: Plot of $\kappa(v,R)$ for $R=1/2$. 
Right: $\text{Rho}_{\overline C}=0$ curves in the $(K/S_0, (r-q)T)$ plane 
for several values of $w = \sigma\sqrt{T}/\sqrt{3}$:
$w=0$ (black), $w=0.5$ (blue) and $w=1$ (red).
The curves shown separate the regions of $\text{Rho}_{\overline C} <0$ (left) and 
$\text{Rho}_{\overline C} >0$ (right). The dashed curve shows the ATM point $K=A(T)$.}
\label{Fig:1}
 \end{figure}

Recall that $\Sigma_{\rm LN}(K/S_0)$ is linear in $\sigma$. Neglecting
the $K/S_0$ dependence in this function, one can set to a good approximation
$\Sigma_{\rm LN}(K/S_0)\sim \sigma/\sqrt{3}$. Then (\ref{ZeroRho}) gives 
the curves shown in Figure~\ref{Fig:1} (right) for several values of
the parameter $w = \sigma\sqrt{T}/\sqrt{3}$ equal to $0, 0.5, 1$. 
From this figure one can see that Asian call options with strikes around the
ATM region $K\sim A(T)$
have positive Rho, unless the unannualized volatility $w$ takes very large values,
above $\sim 100\%$. As the volatility increases, Rho decreases to zero, 
and may become even negative for extremely large volatility.

\subsection{Rho and Psi for Floating-strike Asian options}\label{ctsFloating}

Next, let us consider the floating-strike Asian options.
The prices of the floating-strike Asian options are given by
\begin{eqnarray}\label{floatC}
C_f(\kappa,T,\sigma,r,q) &=& e^{-rT}\mathbb{E}\left[\left(\kappa S_{T}-\frac{1}{T}
     \int_{0}^{T}S_{t}dt\right)^{+}\right]\,,\\
\label{floatP}
P_f(\kappa,T,\sigma,r,q) &=& e^{-rT}\mathbb{E}\left[\left(\frac{1}{T}
     \int_{0}^{T}S_{t}dt-\kappa S_{T}\right)^{+}\right],
\end{eqnarray}
where $\kappa>0$ is a constant, called the strike of the option.

\begin{proposition}\label{prop:2}
In the Black-Scholes model, we have the following properties for the
Rho and Psi of floating-strike Asian options:

(i) The Rho for the floating-strike Asian put options is always negative.

(ii) The Rho for the floating-strike Asian call options is always positive.

(iii) The Psi for the floating-strike Asian put options is sometimes positive
and sometimes negative.

(iv) The Psi for the floating-strike Asian call options is always negative.
\end{proposition}

\begin{proof}
It is easy to observe that
\begin{equation}
C_f=\mathbb{E}\left[\left(\kappa S_{0}e^{-qT-\frac{1}{2}\sigma^{2}T+\sigma W_{T}}
-\frac{1}{T}\int_{0}^{T}S_{0}e^{r(t-T)-qt-\frac{1}{2}\sigma^{2}t+\sigma W_{t}}dt\right)^{+}\right],
\end{equation}
and
\begin{equation}
P_f=e^{-rT}\mathbb{E}\left[\left(\frac{1}{T}\int_{0}^{T}S_{0}e^{r(t-T)-qt-\frac{1}{2}\sigma^{2}t+\sigma W_{t}}dt
-\kappa S_{0}e^{-qT-\frac{1}{2}\sigma^{2}T+\sigma W_{T}}\right)^{+}\right].
\end{equation}
Hence, we conclude that Rho is positive for the floating-strike Asian
call option, and 
negative for the floating-strike Asian put option. This proves (i) and (ii).

A more direct proof can be given using the equivalence relation 
between fixed-strike and floating-strike Asian options in the BS model noted
in \cite{HW}. It was shown in this paper that the prices of floating-strike
Asian options can be expressed as
\begin{align}
&C_f(\kappa, T,\sigma,r,q)=e^{-qT} \mathbb{E}^*[(\kappa S_0 - A_T)^+] = 
 P(\kappa S_0,T,\sigma,q,r)\,, \\
&P_f(\kappa, T,\sigma,r,q)=e^{-qT} \mathbb{E}^*[(A_T - \kappa S_0)^+] = 
 C(\kappa S_0, T,\sigma, q, r)\,,
\end{align}
where the expectations $\mathbb{E}^*$ on the right-hand side are taken with respect to a 
measure where the asset price follows the diffusion 
$dS_t^* = (q-r) S_t^* dt + \sigma S_t^* dW_t^*$.
This relates the prices of call/put floating-strike Asian options to those
of put/call fixed-strike Asian options on the asset $S_t^*$. We also have the
relation among Greeks of these options: for floating-strike call Asian
option 
\begin{equation}\label{RhoPsi}
\text{Rho}_{f,C}(\kappa;r,q) = \Psi_P(\kappa;q,r)\,,\qquad
\Psi_{f,P}(\kappa;r,q) = \text{Rho}_C(\kappa;q,r)\,,
\end{equation}
where the subscripts on the right side refer to put/call fixed-strike Asian
options, respectively. Using these relations, the results of 
Proposition~\ref{prop:2}
follow from (iii) and (iv) of Proposition~\ref{prop:1}.

In a similar way, (i) and (ii) of Proposition~\ref{prop:1} give the results
(iii) and (iv) of Proposition~\ref{prop:2}.
\end{proof}


\appendix

\section{An alternative approach using Malliavin calculus and large deviations}\label{ctsAlt}

Using Malliavin calculus methods, \cite{BP} showed that the
Rho of an Asian call option in the BS model has an alternative expression 
that is given by
\begin{equation}
\text{Rho}_C = e^{-rT}
\mathbb{E}\left[\left(\frac{1}{T}\int_{0}^{T}S_{t}dt-K\right)^{+}\left(\frac{W_{T}}{\sigma}-T\right)\right]
=J_{1}-J_{2},
\end{equation}
where
\begin{align*}
&J_{1}:=e^{-rT}\mathbb{E}\left[\left(\frac{1}{T}\int_{0}^{T}S_{t}dt-K\right)^{+}\frac{W_{T}}{\sigma}\right],
\\
&J_{2}:=e^{-rT}\mathbb{E}\left[\left(\frac{1}{T}\int_{0}^{T}S_{t}dt-K\right)^{+}T\right].
\end{align*}
We can compute that
\begin{align*}
\mathbb{E}\left[\left(\frac{1}{T}\int_{0}^{T}S_{t}dt-K\right)^{+}\right]
&=\mathbb{E}\left[\left(\frac{1}{T}\int_{0}^{T}S_{t}dt-K\right)\right]
+\mathbb{E}\left[\left(K-\frac{1}{T}\int_{0}^{T}S_{t}dt\right)^{+}\right]
\\
&=\frac{1}{T}\int_{0}^{T}S_{0}e^{(r-q)t}dt-K
+\mathbb{E}\left[\left(K-\frac{1}{T}\int_{0}^{T}S_{t}dt\right)^{+}\right].
\end{align*}
Assume that the Asian call option is in-the-money, i.e.
\begin{equation}\label{ITMAssump}
\frac{1}{T}\int_{0}^{T}S_{0}e^{(r-q)t}dt-K>0.
\end{equation}
Then, as $\sigma\rightarrow 0$, by bounded convergence theorem, 
\begin{equation}
\mathbb{E}\left[\left(K-\frac{1}{T}\int_{0}^{T}S_{t}dt\right)^{+}\right]\rightarrow 0.
\end{equation}
Hence,
\begin{equation}
\lim_{\sigma\rightarrow 0}J_{2}=e^{-rT}\left(\frac{1}{T}\int_{0}^{T}S_{0}e^{(r-q)t}dt-K\right)T.
\end{equation}

On the other hand,
\begin{align*}
\mathbb{E}\left[\left(\frac{1}{T}\int_{0}^{T}S_{t}dt-K\right)^{+}W_{T}\right]
&=\mathbb{E}\left[\left(\frac{1}{T}\int_{0}^{T}S_{t}dt-K\right)W_{T}\right]
+\mathbb{E}\left[\left(K-\frac{1}{T}\int_{0}^{T}S_{t}dt\right)^{+}W_{T}\right]
\\
&=\mathbb{E}\left[\frac{1}{T}\int_{0}^{T}S_{t}dt\cdot W_{T}\right]
+\mathbb{E}\left[\left(K-\frac{1}{T}\int_{0}^{T}S_{t}dt\right)^{+}W_{T}\right].
\end{align*}
The second term is bounded by
\begin{align*}
&\left|\mathbb{E}\left[\left(K-\frac{1}{T}\int_{0}^{T}S_{t}dt\right)^{+}W_{T}\right]\right|
\\
&\leq\left(\mathbb{E}\left[\left(K-\frac{1}{T}\int_{0}^{T}S_{t}dt\right)^{2}1_{K>\frac{1}{T}\int_{0}^{T}S_{t}dt}\right]\right)^{1/2}
\left(\mathbb{E}[W_{T}^{2}]\right)^{1/2}
\\
&\leq
\left(\mathbb{E}\left[\left(K-\frac{1}{T}\int_{0}^{T}S_{t}dt\right)^{4}\right]\right)^{1/4}
\mathbb{P}\left(K>\frac{1}{T}\int_{0}^{T}S_{t}dt\right)^{1/4}
\left(\mathbb{E}[W_{T}^{2}]\right)^{1/2}.
\end{align*}

Notice that we have
\begin{equation}
\mathbb{P}\left(K>\frac{1}{T}\int_{0}^{T}S_{t}dt\right)
=
\mathbb{P}\left(K>\frac{1}{T}\int_{0}^{T}S_{0}e^{(r-q)t-\frac{1}{2}\sigma^{2}t+\sigma W_{t}}dt\right),
\end{equation}
for any $\sigma$ sufficiently small.

From the assumption in \eqref{ITMAssump}, for sufficiently small $\epsilon>0$,
\begin{equation}
K<\frac{1}{T}\int_{0}^{T}S_{0}e^{(r-q)t-\epsilon t}dt.
\end{equation}

Notice that we have
\begin{align*}
\mathbb{P}\left(K>\frac{1}{T}\int_{0}^{T}S_{t}dt\right)
&=
\mathbb{P}\left(K>\frac{1}{T}\int_{0}^{T}S_{0}e^{(r-q)t-\frac{1}{2}\sigma^{2}t+\sigma W_{t}}dt\right)
\\
&\leq
\mathbb{P}\left(K>\frac{1}{T}\int_{0}^{T}S_{0}e^{(r-q)t-\epsilon t+\sigma W_{t}}dt\right)\,,
\end{align*}
for any $\sigma$ sufficiently small.

As $\sigma\rightarrow 0$, 
Schilder's theorem from large deviations theory (see e.g. \cite{LD,VarLD})
says that as $\sigma\rightarrow 0$, $\mathbb{P}(\sigma W_{t}\in\cdot)$ satisfies
a large deviation principle on $C_{0}[0,T]$ (the space of continuous functions starting at $0$
at time $0$ equipped with uniform topology) with speed $1/\sigma^{2}$ and the good rate function
\begin{equation}
\frac{1}{2}\int_{0}^{T}[f'(t)]^{2}dt
\end{equation}
if $f\in\mathcal{AC}_{0}[0,T]$ (the space of absolutely continuous functions starting
at $0$ at time $0$) and $+\infty$ otherwise.
Therefore, by contraction principle (see e.g. \cite{LD,VarLD})
\begin{align*}
&\lim_{\sigma\rightarrow 0}\sigma^{2}
\log\mathbb{P}\left(K>\frac{1}{T}\int_{0}^{T}S_{0}e^{(r-q)t-\epsilon t+\sigma W_{t}}dt\right)
\\
&=-\inf_{f\in\mathcal{AC}_{0}[0,T]:K>\frac{1}{T}\int_{0}^{T}S_{0}e^{(r-q)t-\epsilon t+f(t)}dt}
\frac{1}{2}\int_{0}^{T}[f'(t)]^{2}dt.
\end{align*}
Hence, we conclude that, as $\sigma\rightarrow 0$,
\begin{equation}
\mathbb{P}\left(K>\frac{1}{T}\int_{0}^{T}S_{t}dt\right)=e^{-O(\frac{1}{\sigma^{2}})}.
\end{equation}

Next, let us compute 
\begin{equation}
\mathbb{E}\left[\frac{1}{T}\int_{0}^{T}S_{t}dt\cdot W_{T}\right]
=\frac{1}{T}\mathbb{E}\left[\int_{0}^{T}S_{0}e^{(r-q)t-\frac{1}{2}\sigma^{2}t+\sigma W_{t}}dt\cdot W_{T}\right].
\end{equation}
We can expand as $\sigma\rightarrow 0$ to get
\begin{align*}
\mathbb{E}\left[\frac{1}{T}\int_{0}^{T}S_{t}dt\cdot W_{T}\right]
&=\frac{1}{T}\mathbb{E}\left[\int_{0}^{T}S_{0}e^{(r-q)t}dt\cdot W_{T}\right]
+\sigma\frac{1}{T}\mathbb{E}\left[\int_{0}^{T}S_{0}e^{(r-q)t}W_{t}W_{T}dt\right]+O(\sigma^{2})
\\
&=\sigma\frac{1}{T}\int_{0}^{T}S_{0}e^{(r-q)t}tdt+O(\sigma^{2}).
\end{align*}

Hence, we conclude that
\begin{equation}
\lim_{\sigma\rightarrow 0}J_{1}=e^{-rT}\frac{1}{T}\int_{0}^{T}S_{0}e^{(r-q)t}tdt.
\end{equation}

If $K$ is chosen so that \eqref{ITMAssump}
is satisfied and at the same time the quantity
\begin{equation}
\left(\frac{1}{T}\int_{0}^{T}S_{0}e^{(r-q)t}dt-K\right)
\end{equation}
is sufficiently small. Then for sufficiently small volatility $\sigma$, 
we must have $J_{1}>J_{2}$, and thus Rho$_C$ is positive for the choice of such parameters.

\section*{Acknowledgements}
The authors are grateful to an anonymous referee and the editor for helpful
comments and suggestions that have greatly improved the quality of the paper.
Lingjiong Zhu is grateful to the support from NSF Grant DMS-1613164.
The views and opinions expressed in this article are those of the authors, and do
not necessarily represent those of authors' employers.


\end{document}